\documentclass[letterpaper, 10 pt, conference]{ieeeconf}  

\IEEEoverridecommandlockouts                              
\usepackage{latexsym,mathrsfs}
\usepackage{bm}
\usepackage[dvips]{graphics}
\usepackage{graphicx}
\usepackage{psfrag}
\usepackage{amsfonts,amsmath,amssymb}
\usepackage{algorithm,algorithmic}
\usepackage{dsfont,color,subfigure}

\DeclareMathOperator*{\argmin}{arg\,min}

\def\a{\alpha}
\def\b{\beta}
\def\g{\gamma}

\def\e{\epsilon}

\newcommand{\Rr}{\mathscr{R}}
\newcommand{\Cc}{\mathcal{C}}
\newcommand{\Xc}{\mathcal{X}}
\newcommand{\Yc}{\mathcal{Y}}
\newcommand{\Wc}{\mathcal{W}}
\newcommand{\Zc}{\mathcal{Z}}
\newcommand{\Sc}{\mathcal{S}}
\newcommand{\Ec}{\mathcal{E}}
\DeclareMathOperator\E{E}
\let\P\relax
\DeclareMathOperator\P{P}

\newtheorem{theorem}{Theorem}

\begin{document}

\title{\LARGE \bf Multiple Description Coding of Discrete Ergodic Sources}

\author{Shirin Jalali and Tsachy Weissman 
\thanks{S. Jalali is a postdoctoral fellow at the Center for the Mathematics of Information, California Institute of Technology,  Pasadena, CA  91125, USA        
{\tt\small shirin@caltech.edu}}%
\thanks{T. Weissman is with the Department of Electrical Engineering, Stanford  University,  Stanford, CA 94305, USA
        {\tt\small tsachy@stanford.edu}}%
}

\maketitle

\newcommand{\p}{\mathds{P}}
\newcommand{\mb}{\mathbf{m}}
\newcommand{\bb}{\mathbf{b}}

\begin{abstract}

We investigate the problem of Multiple Description (MD) coding of
discrete ergodic processes. We introduce the notion of MD
stationary coding, and characterize its relationship to the
conventional block MD coding. In stationary coding, in addition
to the two rate constraints normally considered in the MD
problem, we consider another rate constraint which reflects the
conditional entropy of the process generated by the third
decoder given the reconstructions of the two other decoders.
The relationship that we establish between stationary and block
MD coding enables us to devise  a universal algorithm for MD
coding of discrete ergodic sources, based on simulated
annealing ideas that were recently proven useful for the
standard rate distortion problem.

\end{abstract}

\section{INTRODUCTION}

Consider a packet network where a signal is to be described to
several receivers. In a basic setup, the source is coded by a
lossy encoder, and several copies of the packet containing the
source description is sent over the network to make sure that
each receiver gets at least one copy. Receiving more than one
copy of these packets is not advantageous, because all the
packets contain similar information. In contrast to this setup,
one can think of a more reasonable scenario where the packets
flooded into the network are not exactly the same; They are
designed such that receiving each one of them is sufficient for
recovering the source, but receiving more packets improves the
quality of the reconstructed signal. The described scenario is
referred to as {\it multiple description}.

The information-theoretic statement of the MD problem, and
early results on the MD problem can be found in
\cite{MD_Witsenhausen}-\cite{MD_WW}. Even for the seemingly
simple case where there are only two receivers, and the source
is i.i.d., the characterization of the achievable
rate-distortion region is not known in general. For this case,
there are two well-known inner bounds due to El Gamal-Cover
\cite{ElGamalCover_MD} and Zhang-Berger \cite{ZhangBerger_MD}.
There is also a combined region, introduced in
\cite{Goyal_combined_MD}, which includes both regions, but
recently shown to be no better than the Zhang-Berger region
\cite{Lei}. In any case, full characterization of the
achievable region is not yet known.

Since even for i.i.d.~sources, the single-letter
characterization of the achievable rate-distortion region is
not known in general, there are few works done on the MD of
non-i.i.d.~sources. The rate-distortion region of Gaussian
processes is derived in \cite{process_MD_dcc07}, and is shown
to be achievable using a scheme based on transform lattice
quantization. In \cite{process_MD_Effros}, a multi-letter
characterization of the achievable weighted rate-distortion
region of discrete stationary ergodic sources is derived.

In this paper, we consider the MD of discrete ergodic processes
where the distribution of the source is not known to the
encoder and decoder. We introduce a universal algorithm which
can asymptotically achieve any point in the achievable
rate-distortion region. In order to get this result, we start
by defining two notions of MD coding, namely, (i) conventional
block coding, and (ii) {\it stationary} coding. In the normal
block-coding MD, there are two rates but three reconstruction
processes. In the stationary coding setup, there are three
rates and three reconstruction processes. The additional rate
corresponds to the {\it conditional entropy rate} of the the
ergodic process reconstructed by the privileged decoder, which
receives two descriptions of the source, given the two other
ergodic reconstruction processes. We show that these two setups
are closely related and, in fact, characterize each other. The
beneficial point of the new definition is that it enables us to
devise a universal MD algorithm. The introduced algorithm takes
advantage of simulated annealing which was used recently in
\cite{JW_r_d} to design an asymptotically optimal universal
algorithm for lossy compression of discrete ergodic sources.

The outline of this paper is as follows: In Section \ref{sec:
notation} some preliminary notation, and definitions are
presented. Section \ref{sec: simple ex} studies a simple
example, which, as made clear later, is closely related to the
MD problem. Section \ref{sec: MD} formally defines the MD
problems, and the two notions of block MD coding and stationary
MD coding, and shows the relationship between the two. Based on
these results, a universal MD algorithm is described in Section
\ref{sec: universal MD}, and in Section \ref{sec:simulate} some simulation results demonstrating the performance of the proposed algorithm on simulated data are presented. Finally, Section \ref{sec: conclusion} discusses some future research directions.

\section{NOTATION}\label{sec: notation}

Let $\mathbf{X}=\{X_i;\forall\; i\in\mathds{N}^{+}\}$ be a
stochastic process defined on a probability space
$(\mathbf{X},\Sigma,\mu)$, where $\mu$ is
a probability measure defined on $\Sigma$, the
$\sigma$-algebra generated by the cylinder sets $\Cc$. For a process
$\mathbf{X}$, let $\Xc$ denote the alphabet set of $X_i$, which is
assumed to be finite in this paper. The shift operator
$T:\Xc^{\infty}\to\Xc^{\infty}$ is defined by
\[
(T\mathbf{x})_n=x_{n+1},\quad\mathbf{x}\in\Xc^{\infty},n\geq1.
\]
Moreover, for a stationary process $\mathbf{X}$, let
$\bar{H}(\mathbf{X})$ denote its entropy rate defined as
$\bar{H}(\mathbf{X})=\lim\limits_{n\to\infty}H(X_{n+1}|X^n)$.

Let $\Xc$ and $\hat{\Xc}$ denote the source and reconstruction
alphabets respectively. For $y^n\in\Yc^n$, define the matrix
$\mathbf{m}(y^n)$ to be the $|\Yc|\times|\Yc|^{k}$ matrix representing
the $(k+1)^{\rm th}$ order empirical distribution of $y^n$, i.e.,
its $(\b, \bb)^{\rm th}$ element is  defined as
\begin{equation}\label{eq: empirical count matrix}
m_{\b,\bb}(y^n) = \frac{1}{n} \left| \left\{ 1 \leq i \leq n :
y_{i-k}^{i-1} = \bb, y_i=\b]    \right\}\right|,
\end{equation}
where $\bb\in\Yc^k$, and $\beta\in\Yc$. In \eqref{eq: empirical
count matrix} and throughout we assume a cyclic convention
whereby $y_i \triangleq y_{n+i}$ for $i \leq 0$. Let $H_k
(y^n)$ denote the conditional empirical entropy of order $k$
induced by $y^n$, i.e.~
\begin{equation}\label{eq: emp cond distribution}
   H_k (y^n) = H(Y_{k+1} | Y^{k}) ,
\end{equation}
where $Y^{k+1}$ on the right hand side of (\ref{eq: emp cond
distribution}) is distributed according to
\begin{equation}\label{eq: empirical distribution}
   \P (Y^{k+1} = [\bb,\b]) = m_{\b,\bb}(y^n).
\end{equation}
The conditional empirical entropy in \eqref{eq: emp cond
distribution} can be expressed as a function of $\mb(y^n)$ as
follows
\begin{equation}\label{eq: alternative representation of Hk}
H_k (y^n) = \frac{1}{n} \sum_{\bb} \mathcal{H} \left(
\mb_{\cdot,\bb}(y^n) \right) \mathbf{1}^T \mb_{\cdot,\bb}(y^n),
\end{equation}
where $\mathbf{1}$ and  $\mb_{\cdot,\bb}(y^n)$ denote the
all-ones column vector of length $|\Yc|$, and the column in
$\mb(y^n)$ corresponding to $\bb$ respectively. For a vector
$\mathbf{v}= (v_1, \ldots , v_\ell)^T$ with non-negative
components, we let $\mathcal{H}(\mathbf{v})$ denote the entropy
of the random variable whose probability mass function (pmf) is
proportional to $\mathbf{v}$. Formally,
\begin{equation}\label{eq: single letter ent functional defined}
\mathcal{H} (\mathbf{v}) = \left\{ \begin{array}{cc}
                           \sum_{i=1}^\ell \frac{v_i}{\| \mathbf{v}
\|_1}  \log \frac{\| \mathbf{v} \|_1}{v_i} &  \mbox{ if }  \mathbf{v}
\neq (0, \ldots , 0)^T \\
                           0 & \mbox{ if } \mathbf{v}  = (0, \ldots , 0)^T.
                         \end{array}
\right.
\end{equation}

Let $\mathbf{m}(w^n|y^n,z^n)$ denote  the conditional $k^{\rm
th}$ order empirical distribution of $w^n$ given $y^n$ and
$z^n$, whose $(\b, \bb_0, \bb_1, \bb_2)^{\rm  th}$ element is
defined as
\begin{align}
& m_{\beta,\mathbf{b}_0,\mathbf{b}_1,\mathbf{b}_2} =\nonumber\\
& \frac{1}{n}\left|\left\{i: w_i=\beta, w_{i-k}^{i-1}=\mathbf{b}_0,y_{i-k_1}^{i+k_1}=\mathbf{b}_1,z_{i-k_1}^{i+k_1}=\mathbf{b}_2\right\}\right|,
\end{align}
where $\beta\in\Wc$, $\mathbf{b}_0\in\Wc^{k}$, $\mathbf{b}_1\in\Yc^{2k_1+1}$, and $\mathbf{b}_2\in\Zc^{2k_1+1}$.
Now define the conditional empirical entropy of $w^n$ given
$y^n$ and $z^n$, $H_{k,k_1}(y^n|w^n,z^n)$, in terms of
$\mathbf{m}(w^n|y^n,z^n)$ as
\begin{align}
H_{k,k_1}(w^n|y^n,z^n)
&=\sum\limits_{\mathbf{b}_0,\mathbf{b}_1,\mathbf{b}_2}\mathbf{1}^T
\mb_{\cdot,\mathbf{b}_0,\mathbf{b}_1,\mathbf{b}_2} \mathcal{H}
\left( \mb_{\cdot,\mathbf{b}_0,\mathbf{b}_1,\mathbf{b}_2}  \right).
\end{align}

\begin{figure}[h]
\begin{center}
\psfrag{R1}[r]{$R_1$ bits}
\psfrag{R2}[r]{$R_2$ bits}
\psfrag{S}[r]{$(S_1,S_2,S_0)$}
\psfrag{S1}[l]{$\hat{S_1}$}
\psfrag{S2}[l]{$\hat{S_2}$}
\psfrag{S0}[l]{$\hat{S_0}$}
\includegraphics[width=3.5cm]{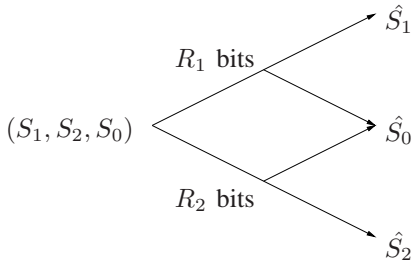}\caption{Example setup}\label{fig:3}
\end{center}
\end{figure}

\section{SIMPLE EXAMPLE} \label{sec: simple ex}

Before formally defining the MD problem, consider the setup
shown in Fig.~\ref{fig:3}. This example is meant to provide
some insight into the MD problem.  Also, the results of this
section will be used in the proof of Theorem \ref{thm: 1} in
Appendix A. Here $S_1\in\Sc_1$, $S_2\in\Sc_2$ and $S_0\in\Sc_0$
denote three correlated discrete-valued random variables, and
$(S_1,S_2,S_0)\sim \P(s_1,s_2,s_0)$. The Encoder's goal is to
send $R_1$ bits to Decoder $1$, and $R_2$ bits to Decoder $2$
such that Decoder $1$ and $2$ are able to reconstruct $S_1$ and
$S_2$ respectively. Moreover, the transmitted bits are required
to be such that receiving both of them enables Decoder $0$ to
reconstruct $S_0$. In all three cases, the probability of error
is assumed to be zero. Let $M_1\in\{1,\ldots,2^{R_1}\}$, and
$M_2\in\{1,\ldots,2^{R_2}\}$ denote the messages sent to the
decoders $1$ and $2$ respectively. The question is to find the
set of achievable rates $(R_1,R_2)$. The following theorem
states some necessary conditions for $(R_1,R_2)$ to be
achievable.  It is very similar to Theorem 2 of \cite{ElGamalCover_MD}, and the two theorems are in fact easily seen to prove each other. The version we give here is most suited for our later needs.
\begin{theorem}\label{thm: 0}
For any achievable rate $(R_1,R_2)$ for the setup shown in
Fig.~\ref{fig:3},
\begin{align}
R_1\geq &H(S_1)\nonumber\\
R_2\geq &H(S_2)\nonumber\\
 R_1+R_2\geq &H(S_1)+H(S_2)+H(S_0|S_1,S_2).\label{eq: cond thm 0}
\end{align}
\end{theorem}
\begin{proof}
$R_1 \geq H(M_1)$ and $R_2\geq H(M_2)$ follow from Shannon's lossless coding Theorem. It is also
clear that we should have
\begin{align}
R_1+R_2&\geq H(S_1,S_2,S_0)\nonumber\\
&= H(S_1,S_2)+H(S_0|S_1,S_2). \label{eq: cond on R1 + R2}
\end{align}
But, perhaps somewhat counterintuitively, \eqref{eq: cond on R1
+ R2} is just an outer bound, and is not enough. $R_1+R_2$ in fact satisfies the tighter condition
stated in \eqref{eq: cond thm 0}, as can be seen via the following chain of inequalities:
\begin{align}
R_1+R_2 & \geq H(M_1) + H(M_2),\nonumber\\
		& = H(M_1,S_1) + H(M_2,S_2),\nonumber\\
		& = H(S_1) + H(M_1|S_1) + H(S_2) + H(M_2|S_2),\nonumber\\
		& \geq H(S_1) + H(S_2) +  H(M_1|S_1,S_2) +\nonumber\\
		&\hspace{5mm} H(M_2|S_1,S_2),\nonumber\\
		& \geq H(S_1) + H(S_2) +  H(M_1,M_2|S_1,S_2) ,\nonumber\\
		& \geq H(S_1) + H(S_2) +  H(M_1,M_2,S_0|S_1,S_2), \nonumber\\
		& \geq H(S_1) + H(S_2) +  H(S_0|S_1,S_2) .
\end{align}
\end{proof}

\section{MULTIPLE DESCRIPTION PROBLEM}\label{sec: MD}

Consider the basic setup of MD problem shown in
Fig.~\ref{fig:2}. In this figure, $X^n$ is generated by a
stationary ergodic source $\mathbf{X}$.

\begin{figure}[t]
\begin{center}
\psfrag{X}[c]{$X^n$}
\psfrag{M1}[b]{$M_1$}
\psfrag{M2}[t]{$M_2$}
\psfrag{enc}[c]{\rm Encoder}
\psfrag{dec1}[c]{\rm Decoder 1}
\psfrag{dec2}[c]{\rm Decoder 2}
\psfrag{dec12}[c]{\rm Decoder 0}
\psfrag{X1}[l]{$\hat X_{1}^n$}
\psfrag{X2}[l]{$\hat X_{2}^n$}
\psfrag{X12}[l]{$\hat X_{0}^n$}
\includegraphics[width=7cm]{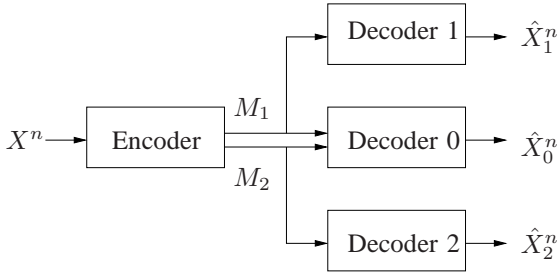}\caption{MD
coding setup}\label{fig:2}
\end{center}
\end{figure}

\textbf{Remark:} In order to see the connection between the
example described in Section \ref{sec: simple ex}, and the MD
problem, note that letting $S_i=\hat{X}_i^n$, $i\in\{1,2\}$,
and $S_0=\hat{X}_0^n$, the MD problem can be described as the
problem of describing $(S_1,S_2,S_0)$ to the respected
receivers error-free. In other words, for each code design, we
have a problem equivalent to the one described in Section
\ref{sec: simple ex}.

\subsection{Block coding:}

MD coding problem can be described in terms of encoding
mapping $f$, and decoding mappings $(g_1,g_2,g_{0})$
as follows

\begin{enumerate}
\item $f:\Xc^n\to[1:2^{nR_1}]\times[1:2^{nR_2}]$,
\item $g_i:[1:2^{nR_i}]\to\hat{\Xc}^n$, for $i=1,2$,
\item $g_{0}:[1:2^{nR_1}]\times[1:2^{nR_2}]\to\hat{\Xc}^n$,
\item $(M_1,M_2)=f(X^n)$,
\item $\hat{X}^n_i=g_i(M_i)$, for $i=1,2$,
\item $\hat{X}^n_0=g_{0}(M_1,M_2)$.
\end{enumerate}

$(R_1,R_2,D_1,D_2,D_{0})$ is said to be achievable for this
setup,
if there exists a sequence of codes \\
$(f^{(n)},g^{(n)}_1,g^{(n)}_2,g^{(n)}_{0})$ such
that
\begin{align*}
\limsup\limits_{n}\E d_n(X^n,\hat{X}^n_i)&\leq D_i, \textmd{ for } \; i=1,2,\\
\limsup\limits_{n}\E d_n(X^n,\hat{X}^n_0)&\leq D_{0}.
\end{align*}

Let $\Rr^{\rm B}$  be the set of all $(R_1,R_2,D_1,D_2,D_{0})$
that are  achievable by block MD coding of source $
\mathbf{X}$.

\subsection{Stationary coding:}

Define $(R_{11},R_{22},R_0,D_1,D_2,D_0)$ to be achievable by
stationary coding of source $\mathbf{X}$, if for any $\e>0$, there exist
processes $\hat{\mathbf{X}}_1^{(\e)}$, $\hat{\mathbf{X}}^{(\e)}_2$ and
$\hat{\mathbf{X}}_0^{(\e)}$ such that
$(\mathbf{X},\hat{\mathbf{X}}^{(\e)}_1,\hat{\mathbf{X}}^{(\e)}_2,\hat{\mathbf{X}}^{(\e)}_{0})$
are jointly stationary ergodic processes, and
\begin{align}
\bar{H}(\hat{\mathbf{X}}^{(\e)}_1) &\leq R_{11}+\e\label{eq:1}\\
\bar{H}(\hat{\mathbf{X}}^{(\e)}_2) &\leq R_{22}+\e\label{eq:2}\\
\bar{H}(\hat{\mathbf{X}}^{(\e)}_{0}|\hat{\mathbf{X}}^{(\e)}_1,\hat{\mathbf{X}}^{(\e)}_2)+
&\leq R_{0}+\e\label{eq:3}\\
\E d(X_0,\hat{X}^{(\e)}_{1,0}) &\leq D_1+\e\label{eq:1-2}\\
\E d(X_0,\hat{X}^{(\e)}_{2,0}) &\leq D_2+\e\label{eq:2-2}\\
\E d(X_0,\hat{X}^{(\e)}_{0,0}) &\leq D_{0}+\e\label{eq:3-2}.
\end{align}

Let $\Rr^{\rm P}$  denote the set of all
$(R_{11},R_{22},R_0,D_1,D_2,D_0)$  that are achievable by
stationary MD coding of source $\mathbf{X}$. The following
theorem characterizes $\Rr^{\rm B}$ in terms of $\Rr^{\rm P}$ .

\begin{theorem} \label{thm: 1} Let $\mathbf{X}$ be a stationary ergodic source.
For any $(R_1,R_2,D_1,D_2,D_0)\in\Rr^{\rm B}$, 
there exists  $(R_{11},R_{22},R_0,D_1,D_2,D_0)\in \Rr^{\rm P}$ such that 
\begin{align}
R_{11}&\leq R_1\label{eq: R_1 in terms of R11}\\
R_{22}&\leq R_2\label{eq: R_2 in terms of R22}\\
R_{11}+R_{22}+R_0&\leq R_1+R_2 \label{eq: R1+R2 in terms of R11 and R22 and R0}
\end{align}

On the other hand, if $(R_{11},R_{22},R_0,D_1,D_2,D_0)\in
\Rr^{\rm P}$, any point $(R_1,R_2,D_1,D_2,D_0)$
satisfying \eqref{eq: R_1 in terms of R11}-\eqref{eq: R1+R2 in
terms of R11 and R22 and R0} belongs to $\Rr^{\rm B}$.
\end{theorem}

\begin{proof}
Refer to Appendix A for an outline of the proof.
\end{proof}

{\bf Remark:} The theorem implies that $\Rr^{\rm B}$ can be
characterized as the set of $(R_1,R_2,D_1,D_2,D_0)$ such that
\begin{align*}
\bar{H}(\hat{\mathbf{X}}_1) &\leq R_1\\
\bar{H}(\hat{\mathbf{X}}_2) &\leq R_2\\
\bar{H}(\hat{\mathbf{X}}_1)+\bar{H}(\hat{\mathbf{X}}_2)+\bar{H}(\hat{\mathbf{X}}_{0}|\hat{\mathbf{X}}_1,\hat{\mathbf{X}}_2)
&\leq R_1+R_2,
\end{align*}
for some jointly stationary ergodic processes
$(\mathbf{X},\hat{\mathbf{X}}_1,\hat{\mathbf{X}}_2,\hat{\mathbf{X}}_{0})$
which satisfy \eqref{eq:1-2}-\eqref{eq:3-2}.

\section{UNIVERSAL MULTIPLE DESCRIPTION CODING}\label{sec: universal
MD}

Equipped with the characterization of the achievable region
established in the previous section, we now turn to our
construction of a universal scheme for this problem. Consider
the following MD algorithm for the setup shown in
Fig.~\ref{fig:2}. Let
\begin{align}
&(\hat{x}_1^n,\hat{x}_2^n,\hat{x}_{0}^n)\triangleq \nonumber\\
& \argmin \limits_{(y^n,z^n,w^n)}\left[\g_1 H_k(y^n)+\g_2H_k(z^n)+\g_{0}H_{k,k_1}(w^n|y^n,z^n)\right.\nonumber\\
&\left.+\alpha_1d_n(x^n,y^n)+\alpha_2d_n(x^n,z^n)+\alpha_{0}d_n(x^n,w^n)\right],\label{eq:universal MD encoder}
\end{align}

Assume that $\g_i\geq0$ and $\alpha_i\geq0$, for
$i\in\{0,1,2\}$, are given Lagrangian coefficients. Also,
assume that $k_1\leq k=o(\log n)$ such that $k_1\to\infty$ as
$n\to\infty$.

\begin{theorem}\label{thm: 2}
Let $\mathbf{X}$ be a stationary ergodic process, and
$(\hat{X}_1^n,\hat{X}_2^n,\hat{X}_0^n)$ denote the output of
the above algorithm to input sequence $X^n$. Then,
\begin{align}
&\limsup\limits_n\nonumber\\
&\left[\g_1 H_k(\hat{X}_1^n)+\g_2H_k(\hat{X}_2^n)+\g_0
H_{k,k_1}(\hat{X}_0^n|\hat{X}_1^n,\hat{X}_2^n)+\right.\nonumber\\
&\left.\alpha_1d_n(X^n,\hat{X}_1^n)+\alpha_2d_n(X^n,\hat{X}_2^n)+\alpha_{0}d_n(X^n,\hat{X}_0^n)\right]
\nonumber\\
&=
\min\;[\g_1 R_{11}+\g_2R_{22}+\g_{0}R_0+\alpha_1D_1+\alpha_2D_2+\alpha_{0}D_0]
\label{eq: statement thm2}
\end{align}
almost surely, where the minimization is over all
$(R_{11},R_{22},R_0,D_1,D_2,D_0)\in \Rr^{\rm P}$.
\end{theorem}

The proof of Theorem \ref{thm: 2} is presented in Appendix B.

After finding $(\hat{x}_1^n,\hat{x}_2^n,\hat{x}_{0}^n)$,
$\hat{x}_1^n$ and $\hat{x}_2^n$ will be described to Decoders 1
and 2 respectively using one of the well-known universal
lossless  compression algorithms, e.g.,~Lempel Ziv algorithm.
Then Encoder forms a description of  $\hat{x}_0^n$ conditioned
on knowing $\hat{x}_1^n$ and $\hat{x}_2^n$  using conditional
Lempel Ziv algorithm or some other universal algorithm for
lossless coding with side information \cite{cond_LZ}.  A
portion $0\leq \theta\leq 1$ of these bits will be included in
the message $M_1$ and the rest in message $M_2$.


For finding an approximate  solution of \eqref{eq:universal MD encoder} instead of doing the required exhaustive search
directly, as done in \cite{JW_r_d}, one can employ simulated
annealing \cite{book: Markov chains}. To do this, we assign a
cost to each
$(y^n,z^n,w^n)\in\hat{\Xc}^n\times\hat{\Xc}^n\times\hat{\Xc}^n$
as follows
\begin{align}
\Ec(y^n&,z^n,w^n):=\nonumber\\ 
&\g_1H_k(y^n)+\g_2H_k(z^n)+\g_{0}H_{k,k_1}(w^n|y^n,z^n)\nonumber\\
&+\alpha_1d_n(x^n,y^n)+\alpha_2d_n(x^n,z^n)+\alpha_{0}d_n(x^n,w^n),\nonumber
\end{align}
and then define the Boltzmann probability distribution at
temperature $T=1/{\b}$ as
\begin{align}
p_{\b}(y^n,z^n,w^n):=\frac{1}{Z}e^{-\b\Ec(y^n,z^n,w^n)},\label{eq: Boltzmann dist}
\end{align}
where $Z$ is a normalizing constant. Sampling from this
distribution at a very low temperature  yields
$(\hat{X}_1^n,\hat{X}_2^n,\hat{X}_{0}^n)$ with energy close to
the minimum possible energy, i.e.,
\begin{align}
\Ec(\hat{X}_1^n,\hat{X}_2^n,\hat{X}_{0}^n)\approx\min\limits_{(y^n,z^n,w^n)}\Ec(y^n,z^n,w^n).
\end{align}
Since sampling from \eqref{eq: Boltzmann dist} at low
temperatures is almost as hard as doing the exhaustive search,
we turn to simulated annealing (SA) which is a known method for
solving discrete optimization problems. The SA procedure works
as follows: it first defines Boltzmann distribution over the
optimization space, and then tries to sample from the defined
distribution while gradually decreasing the temperature from
some high $T$ to zero according to a properly chosen
\emph{annealing schedule}.

Given $\Ec(y^n,z^n,w^n)$, similarly as in \cite{JW_r_d}, the
number of computations required for calculating \\
$\Ec(y^{i-1}ay_{i+1}^n, z^{i-1}bz_{i+1}^n, w^{i-1}cw_{i+1}^n)$
, when only one of the following is true: $a\neq y_i$, $b\neq
z_i$, or $c\neq w_i$, for some $i\in\{1,\ldots,n\}$ and
$a,b,c\in\hat{\Xc}$, is linear in $k$ and $k_1$, and is
independent of $n$. Therefore, this energy function lends
itself to a heat bath type algorithm as simply and naturally as
the one in the original setting of \cite{JW_r_d} did.

Now consider Algorithm \ref{alg:mcmc_MD} which is based on the Gibbs sampling method for sampling from $p_\beta$, and let $(\hat{X}^{n}_{1,r},\hat{X}^n_{2,r},\hat{X}^n_{0,r})$ denote its random outcome for the input sequence $X^n$ after $r$ iterations\footnote{Here and throughout it is implicit that the randomness used in the algorithms is independent of the source, and the randomization variables used at each drawing are independent of each other.} , when taking $k_1=k_{1,n}$ , $k=k_n$ and $\beta=\{\beta_t\}_t$ to be  deterministic sequences satisfying $k_{1,n} = o( \log n)$, $k_n = o( \log n)$ such that $k,k_1\to\infty$ as $n\to\infty$, and $\beta_t = \frac{1}{T_0^{(n)}}\log(\lfloor\frac{t}{n}\rfloor+1)$, for some $T_0^{(n)}>n\max(\Delta_1,\Delta_2,\Delta_0)$, where
\begin{align}
\Delta_1= & \max\left|\Ec(y^{i-1}ay_{i+1}^n,z^n,w^n)-\Ec(y^{i-1}by_{i+1}^n,z^n,w^n)\right|,\nonumber\\
		&	i\in\{1,\ldots,n\}\nonumber\\
                   &      y^{i-1} \in \hat{\Xc}^{i-1},    y_{i+1}^n \in \hat{\Xc}^{n-i}, \nonumber\\
                   &      a,b \in \hat{\Xc},\nonumber\\
                   &      z^n\in\hat{\Xc}^n,  w^n\in\hat{\Xc}^n,
\end{align}
\begin{align}
\Delta_2=&\max \left|\Ec(y^n,z^{i-1}az_{i+1}^n,w^n)-\Ec(y^n,z^{i-1}bz_{i+1}^n,w^n)\right|,\nonumber\\
&				i\in\{1,\ldots,n\}\nonumber\\
&                             z^{i-1} \in \hat{\Xc}^{i-1},  z_{i+1}^n \in \hat{\Xc}^{n-i}, \nonumber\\
&                             a,b \in \hat{\Xc},\nonumber\\
&                             y^n\in\hat{\Xc}^n, w^n\in\hat{\Xc}^n,
\end{align}
\begin{align}
\Delta_0= &\max |\Ec(y^n,z^n,w^{i-1}aw_{i+1}^n)-\Ec(y^n,z^n,w^{i-1}bw_{i+1}^n)|.\nonumber\\
			&i\in\{1,\ldots,n\}\nonumber\\
                             &w^{i-1} \in \hat{\Xc}^{i-1}, w_{i+1}^n \in \hat{\Xc}^{n-i},\nonumber \\
                             &a,b \in \hat{\Xc},\nonumber\\
                             &y^n\in\hat{\Xc}^n, z^n\in\hat{\Xc}^n,
\end{align}
As discussed before, the computational complexity of the algorithm at each iteration is independent of $n$ and linear in $k$ and $k_1$. Following exactly the same steps as in the proof of Theorem 2 in \cite{JW_r_d}, we can prove the following theorem which established universal optimality of Algorithm \ref{alg:mcmc_MD}.
\begin{theorem}\label{thm:MCMC_MD}
For any ergodic process $\mathbf{X}$,
\begin{align}
&\lim\limits_{n\to\infty}\lim\limits_{r\to\infty} \Ec(\hat{X}_1^n,\hat{X}_2^n,\hat{X}_0^n)\nonumber\\
&=\min\;[\g_1R_{11}+\g_2R_{22}+\g_{0}R_0+\alpha_1D_1+\alpha_2D_2+\alpha_{0}D_0]
\label{eq: statement thm4}
\end{align}
almost surely, where the minimization is over all
$(R_{11},R_{22},R_0,D_1,D_2,D_0)\in \Rr^{\rm P}(\mathbf{X})$.
\end{theorem}

\begin{algorithm}[h!]
\caption{Generating the reconstruction sequences}
\label{alg:mcmc_MD}
\begin{algorithmic}[1]
\REQUIRE $x^n$, $k_1$, $k$, $\{\a_i\}_{i=0}^2$, $\{\b_i\}_{i=0}^2$ $\{ \beta_t \}_{t=1}^r$, $r$
\ENSURE a reconstruction sequences $(\hat{x}_1^n,\hat{x}_2^n,\hat{x}_0^n)$
\STATE $y^n \leftarrow x^n$
\STATE $z^n \leftarrow x^n$
\STATE $w^n \leftarrow x^n$
\FOR{$t=1$ to $r$}
\STATE Draw an integer $i \in \{1, \ldots, n \}$ uniformly at random
\STATE For each $y \in \hat{\Xc}$ compute $q_1(y)=p_{\beta_t} (Y_i = y | Y^{n \setminus i} = y^{n \setminus i},Z^n=z^n,W^n=w^n)$ 
\STATE Update $y^n$ by letting $y_i=V_1$, where $V_1\sim q_1$
\STATE For each $z \in \hat{\Xc}$ compute $q_2(z)=p_{\beta_t} (Z_i = z | Y^n=y^n,Z^{n \setminus i} = z^{n \setminus i},W^n=w^n)$ 
\STATE Update $z^n$ by letting $z_i=V_2$, where $V_2\sim q_2$
\STATE For each $y \in \hat{\Xc}$ compute $p_{\beta_t} (Y_i = y | Y^{n \setminus i} = y^{n \setminus i})$ 
\STATE Update $w^n$ by letting $w_i=V_0$, where $V_0\sim q_0$
\STATE Update $\mb(y^{n})$, $\mb(z^{n})$ and $\mb(w^{n}|y^n,z^n)$
\ENDFOR \STATE $\hat{x}^n \leftarrow y^{n}$
\end{algorithmic}
\end{algorithm}

\section{SIMULATION RESULTS}\label{sec:simulate}

In this section, we present some results showing the actual implementation of the algorithm described in Section \ref{sec: universal MD}. The simulated source here is a sym metric binary Markov source with transition probability $p=0.2$. The considered block length is $n=10^4$, and the context sizes are $k = 5$ and $k_1=1$. The annealing schedule was chosen according to
\[T(t)=\frac{1}{2nt^{\small{1/10}}},\]
where $t$ is the iteration number. The number of iterations, $r$, is equal to $50n$. The algorithm with the specified parameters, for $\g_1=\g_2=\g_0=\a_1=\a_2=a_0 = 1$, achieves the following set of rates and distortions:
\begin{align*}
H_k(\hat{x}^n_1) &=  0.5503, \nonumber\\
H_k(\hat{x}^n_2) &= 0.5586, \nonumber\\
H_{k,k_1}(\hat{x}^n_0|\hat{x}^n_1,\hat{x}^n_2) &= 0.0038,\nonumber\\
d_n(x^n,\hat{x}^n_1) &= 0.0505, \nonumber\\ 
d_n(x^n,\hat{x}^n_2) &=  0.0483,\nonumber\\
d_n(x^n,\hat{x}^n_0) &= 0.0036.\nonumber
\end{align*}
Fig.~\ref{fig: cost reduction} shows how the total cost is reducing in this case, as the number of iterations increases. One interesting thing to note here is that although the sequences $\hat{x}^n_1$ and $\hat{x}^n_2$ have almost the same distance from the original sequence $x^n$, they are far from being equal. In fact, $d_n(\hat{x}^n_1,\hat{x}^n_2)=0.0966$, which, given $d_n(x^n,\hat{x}^n_1) = 0.0505$ and $d_n(x^n,\hat{x}^n_2) =  0.0483$, suggests that they are almost maximally distant.

As another example, consider the case where $n=5\times10^4$ and $\a_1=\a_2=2$. The rest of the parameters are left unchanged. The achieved point in this case is going to be
\begin{align*}
H_k(\hat{x}^n_1) &=  0.6091, \nonumber\\
H_k(\hat{x}^n_2) &= 0.5951, \nonumber\\
H_{k,k_1}(\hat{x}^n_0|\hat{x}^n_1,\hat{x}^n_2) &= 0, \nonumber\\
d_n(x^n,\hat{x}^n_1) &= 0.0200, \nonumber\\
d_n(x^n,\hat{x}^n_2) &=  0.0240, \nonumber\\
d_n(x^n,\hat{x}^n_0) &= 0.0010.
\end{align*}
Here, $H_{k,k_1}(\hat{x}^n_0|\hat{x}^n_1,\hat{x}^n_2)=0$ implies that $\hat{x}_{0,i}$ is a deterministic function of its context,  $(\hat{x}_{0,i-k_1}^{i-1},\hat{x}_{1,i-k_1}^{i+k_1},\hat{x}_{2,i-k_1}^{i+k_1})$. This of course does not mean that no additional rate is required for describing $\hat{x}^n_0$ when the decoder already knows $\hat{x}^n_1$ and $\hat{x}^n_2$, because this deterministic mapping itself is not known to the decoder beforehand. Here again $\hat{x}^n_1$ and $\hat{x}^n_2$ are almost maximally distant because $d_n(\hat{x}^n_1,\hat{x}^n_2)=0.0436$.

Note that the fundamental performance limits are unknown even for memoryless sources and, a fortiori, for the Markov source in our experiment. Thus the performance of our algorithm cannot be compared to the corresponding optimum performance. The results of the preceding section, however, imply that our algorithm attains that performance in the limit of many iterations and large block length. Thus, the performance attained by our algorithm, can alternatively be viewed as approximating the unknown optimum.

\begin{figure}
\begin{center}
\includegraphics[width=9.5cm]{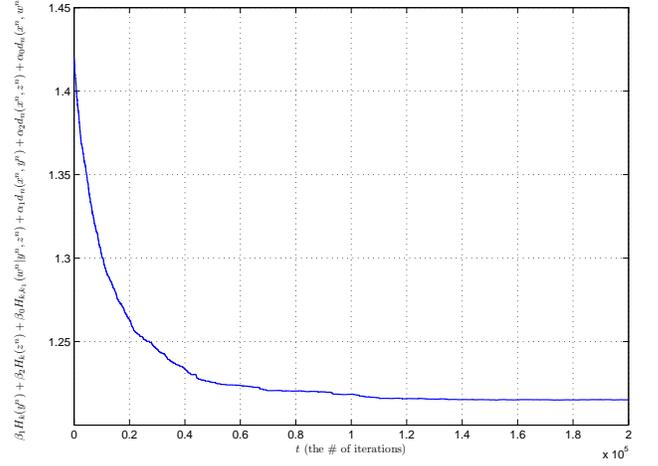}\caption{Reduction in the cost. At the end of the process, the final achived point is: $(H_k(\hat{x}^n_1),H_k(\hat{x}^n_2),H_{k,k_1}(\hat{x}^n_0|\hat{x}^n_1,\hat{x}^n_2),d_n(x^n,\hat{x}^n_1), d_n(x^n,\hat{x}^n_2),$ $d_n(x^n,\hat{x}^n_0))$ $=(0.5503, 0.5586, 0.0038, 0.0505, 0.0483, 0.0036) $}\label{fig: cost reduction}
\end{center}
\end{figure}

\section{FUTURE DIRECTIONS}\label{sec: conclusion}

Simulated annealing was recently employed in \cite{JW_r_d}  to
design a universal lossy compression algorithm. In this paper,
we proved that in fact the same tool can be applied to devise a
universal MD algorithm. We started by defining the equivalent
of MD problem for ergodic processes, and defined the idea of
stationary MD coding which includes three rate constraints
instead of two. Extensions of these results to additional
distributed coding scenarios are under current investigation.

\section*{ACKNOWLEDGMENT}
We thank Jun Chen for suggesting the current proof of Theorem 1, in lieu of our original proof which was more complicated.

\renewcommand{\theequation}{A-\arabic{equation}}
\setcounter{equation}{0}  

\section*{Appendix A: Outline of the proof of Theorem \ref{thm: 1}}
\label{appA}

{\it Outline of the proof of the first part}: Let
$(R_1,R_2,D_1,D_2,D_0)\in\Rr^{\rm B}$. We need to find
$(R_{11},R_{22},R_0)$ such that
$(R_{11},R_{22},R_0,D_1,D_2,D_0)\in\Rr^{\rm P}$, and \eqref{eq:
R_1 in terms of R11} -\eqref{eq: R1+R2 in terms of R11 and R22
and R0} are satisfied.

Let $(f^{(n)},g_1^{(n)},g_2^{(n)},g_0^{(n)})$ be a sequence of
codes at rate $(R_1,R_2)$  that achieves the point
$(R_1,R_2,D_1,D_2,D_0)\in\Rr^{\rm B}$. Note that for a given
code, $(\hat{X}_{1}^n,\hat{X}_{2}^n,\hat{X}_{0}^n)$ is a
deterministic function of $X^n$. Using the same method used in
\cite{Gray_SB}, we can generate jointly stationary ergodic
processes
$(\hat{\mathbf{X}}_1^{(n)},\hat{\mathbf{X}}_2^{(n)},\hat{\mathbf{X}}_0^{(n)})$ by
appropriately embedding these block codes into ergodic
processes.  Here the superscript $(n)$ indicates the dependence of the constructed processes on $n$. 
In order to code an ergodic process into another
ergodic process using a block code of length $n$, we need to
cover an infinite length sequence by non-overlapping blocks of
length $n$  up to a set of negligible measure, and then replace
each block by its reconstruction generated by the block code.
The challenging part is the partitioning which should preserve
the ergodicity. This  can be done using R-K Theorem \cite{R-K
thm} which states that:
\begin{theorem}[Rohlin-Kakutani Theorem] \label{R-K thm}
Given the ergodic source $\textbf{X}$, integers $N$,  $n\leq
N$, and $\epsilon>0$, there exists an event $F$ (called the
\textit{base}) such that
\begin{enumerate}
\item $F,TF,\ldots,T^{N-1}F$ are disjoint,
\item $\P\left(\bigcup\limits_{i=0}^{N-1}T^iF\right)\geq
    1-\epsilon$,
\item
    $\P\left(\mathcal{S}(a^n)|F\right)=\P\left(\mathcal{S}(a^n)\right)$,
    where $\mathcal{S}(a^n)=\{\mathbf{x}: x^n=a^n\}$.
\end{enumerate}
\end{theorem}

Since the sequence of MD block codes were assumed to achieve
the point $(R_1,R_2,D_1,D_2,D_0)$, the constructed process
$(\mathbf{\hat{X}}^{(n)}_1,\mathbf{\hat{X}}^{(n)}_2,\mathbf{\hat{X}}^{(n)}_0)$
satisfies the  distortion constraints given in
\eqref{eq:1-2}-\eqref{eq:3-2} at $(D_1+\e_n, D_2+\e_n,
D_0+\e_n)$, where $\e_n\to0$ as $n\to\infty$. Therefore,
$(\bar{H}_n(\mathbf{\hat{X}}^{(n)}_1),\bar{H}_n(\mathbf{\hat{X}}^{(n)}_2),\bar{H}_n(\mathbf{\hat{X}}^{(n)}_0|\mathbf{\hat{X}}^{(n)}_1,\mathbf{\hat{X}}^{(n)}_2),D_1+\e_n,D_2+\e_n,D_0+\e_n)\in\Rr^{\rm P}$. Let
\begin{align}
R_{11}^{(n)}&:=\frac{1}{n}H(\hat{X}^n_1),\\
R_{22}^{(n)}&:=\frac{1}{n}H(\hat{X}^n_2),\\
R_0^{(n)}&:=\frac{1}{n}H(\hat{X}^n_0|\hat{X}^n_1,\hat{X}^n_2),
\end{align}
where $\hat{X}^n_i=g_i^{(n)}(M_i)$, for $i\in\{1,2\}$ and
$\hat{X}^n_0=g_0^{(n)}(M_1,M_2)$. Note that since the encoder
knows $(\hat{X}^n_1,\hat{X}^n_2,\hat{X}^n_0)$, by Theorem
\ref{thm: 0}, $R_{11}^{(n)}\leq R_1$, $R_{22}^{(n)}\leq R_2$,
and $R_{11}^{(n)}+R_{22}^{(n)}+R_{0}^{(n)}\leq R_1+R_2$. The
only remaining step is to find the relationship between
$(\bar{H}_n(\mathbf{\hat{X}}^{(n)}_1),\bar{H}_n(\mathbf{\hat{X}}^{(n)}_2),\bar{H}_n(\mathbf{\hat{X}}^{(n)}_0|\mathbf{\hat{X}}^{(n)}_1,\mathbf{\hat{X}}^{(n)}_2))$
and $(R_{11}^{(n)},R_{22}^{(n)},R_0^{(n)})$, which is not hard
from the way the processes are constructed.

{\it Outline of the proof of the second part}: Let
$(R_{11},R_{22},R_0,D_1,D_2,D_{0})\in\Rr^{\rm P}$. This means
that there exist processes $\hat{\mathbf{X}}_1$,
$\hat{\mathbf{X}}_2$ and $\hat{\mathbf{X}}_{0}$ jointly
stationary and ergodic with $\mathbf{X}$ which satisfy
\eqref{eq:1}-\eqref{eq:3-2}. Based on these processes, for
block length $n$, we use the following block coding strategy:
For coding sequence $X^n$, describe $\hat{X}_1^n$ and
$\hat{X}_2^n$ losslessly to the decoders 1 and 2 using
$n(\bar{H}(\hat{\mathbf{X}}_1)+\e_n)$ and
$n(\bar{H}(\hat{\mathbf{X}}_2)+\e_n)$ bits respectively. Given
$\hat{X}_1^n$ and $\hat{X}_2^n$,
$n(\bar{H}(\hat{\mathbf{X}}_{0}|\hat{\mathbf{X}}_1,\hat{\mathbf{X}}_2)+\e_n)$
bits suffice to describe $\hat{X}_{0}^n$ losslessly to Decoder
$0$. These bits can be divided into two parts: the first part
will be included in the message $M_1$,  and the rest in the
message $M_2$. Decoders  $1$ and $2$ just ignore these extra
bits, but Decoder $0$ combines them with the two other messages
to reconstruct $\hat{X}_0^n$.  Since $R_1$ and $R_2$ satisfy
\eqref{eq: R_1 in terms of R11}-\eqref{eq: R1+R2 in terms of
R11 and R22 and R0}, it is possible to do this.

\renewcommand{\theequation}{B-\arabic{equation}}
\setcounter{equation}{0}  

\section*{Appendix B: Proof of Theorem \ref{thm: 2}}
\label{appB} 

For an ergodic source $\mathbf{X}$, let
\begin{align}
&\mu(\mathbf{\boldsymbol\g},\boldsymbol\a):=\nonumber\\
&\min\limits_{\Rr^{\rm P}(\mathbf{X})}\left[\g_1 R_{11}+\g_2R_{22}+\g_{0}R_0+\alpha_1D_1+\alpha_2D_2+\alpha_{0}D_0\right].
\end{align}
No coding strategy can beat $\mu(\mathbf{\boldsymbol\g},\boldsymbol\a)$ on a set of non-zero probability.
Therefore, the left hand side of \eqref{eq: statement thm2} is lower bounded by its right hand side. Therefore, we
only need to prove the other direction. By definition, for any
$(R_{11},R_{22},R_0,D_1,D_2,D_0)\in \Rr^{\rm P}(\mathbf{X})$, there exist
processes $\hat{\mathbf{X}}_1$, $\hat{\mathbf{X}}_2$ and
$\hat{\mathbf{X}}_0$ such that \eqref{eq:1}-\eqref{eq:3-2} are
satisfied. On the other hand, if
$(\hat{X}_1^n,\hat{X}_2^n,\hat{X}_0^n)$ is generated by jointly
ergodic processes
$(\hat{\mathbf{X}}_1,\hat{\mathbf{X}}_2,\hat{\mathbf{X}}_0)$,
then for $k=o(\log n)$ and $k_1=o(\log n)$ such that $k,k_1\to\infty$ as $n\to\infty$,
$H_k(\hat{X}_i^n)\to\bar{H}(\hat{\mathbf{X}}_i)$, for
$i\in\{1,2\}$, and moreover
$H_{k,k_1}(\hat{X}_0^n|\hat{X}_1^n,\hat{X}_2^n)\to\bar{H}(\hat{\mathbf{X}}_0|\hat{\mathbf{X}}_1,\hat{\mathbf{X}}_2)$.
This implies that
\begin{align}
\limsup \min[&\g_1 H _k(\hat{X}_1^n)+\a_2d_n(X^n,\hat{X}_2^n)+\nonumber\\
&\g_2H_k(\hat{X}_2^n)+\a_1d_n(X^n,\hat{X}_1^n)+\nonumber\\
&\g_0 H_{k,k_1}(\hat{X}_0^n|\hat{X}_1^n,\hat{X}_2^n)+\alpha_{0}d_n(X^n,\hat{X}_0^n)
]\label{eq: cost}
\end{align}
is upper-bounded by $\mu(\mathbf{\boldsymbol\g},\boldsymbol\a)+\e_n$,
where $\e_n\to 0$.  Combining these two results in  the desired
conclusion.

\end{document}